\documentclass[prl,twocolumn,aps,superscriptaddress,longbibliography]{revtex4-2}

\usepackage{amsmath,amssymb,amsthm}
\usepackage{mathtools}
\usepackage{bm}
\usepackage{graphicx}
\usepackage{hyperref}
\usepackage{xcolor}
\usepackage{cleveref}
\usepackage{tikz}
\usetikzlibrary{arrows.meta, decorations.markings}

\allowdisplaybreaks

\newtheoremstyle{prlthm}{}{}{\itshape}{}{\bfseries}{.}{ }{}
\theoremstyle{prlthm}
\newtheorem{theorem}{Theorem}
\newtheorem{corollary}[theorem]{Corollary}

\newcommand{\tr}{\operatorname{Tr}}
\newcommand{\cF}{\mathcal{F}}

\newcommand{\ket}[1]{|#1\rangle}
\newcommand{\bra}[1]{\langle#1|}

\newcommand{\Var}{\operatorname{Var}}

\begin{document}

\title{Precision and Privacy in Distributed Quantum Sensing:\\
A Quantum Fisher Information Duality}

\author{Farhad Farokhi}
\affiliation{Department of Electrical and Electronic Engineering,
University of Melbourne, Melbourne, VIC 3010, Australia}

\date{\today}

\begin{abstract}
We establish a quantum Fisher information (QFI) duality for distributed quantum sensor networks with local phase encoding. For any $N$-qubit probe state, where $N$ denotes the number of sensors, $\mathcal{F}_Q( \hat{\bm{w}}^\top \bm{\theta}) +\mathcal{F}_Q( \hat{\bm{v}}^\top \bm{\theta})\leq N$ for all unit orthogonal sensing directions $\hat{\bm{w}}$ and $\hat{\bm{v}}$, with equality for all equatorial
states when $N=2$ and for Greenberger–Horne–Zeilinger (GHZ) states when $N\geq 2$. Heisenberg-limited precision for direction $\hat{\bm{w}}$, $\mathcal{F}_Q(\hat{\bm{w}}^\top \bm{\theta})=N$, saturates the bound and simultaneously forces zero QFI for all other independent directions. This can be interpreted as the condition for parameter privacy in distributed quantum sensing: attaining Heisenberg-limited precision for the sensing target renders all alternative privacy-intrusive estimations impossible. 
\end{abstract}

\maketitle

\section{Introduction}
Distributed quantum sensor networks, such as arrays of entangled atomic
clocks~\cite{komar2014quantum}, nitrogen-vacancy (NV) centre magnetometers~\cite{barry2016optical}, and quantum
gravimeters~\cite{stray2022quantum}, are transitioning from laboratory demonstrations to operational infrastructure~\cite{degen2017quantum}. Their defining advantage is Heisenberg-limited precision, that is, distributing an appropriate probe state, such as  Greenberger–Horne–Zeilinger (GHZ)-entangled states, across $N$ sensors reduces the variance of a distributed field estimate~\cite{giovannetti2006quantum,eldredge2018optimal,proctor2018multiparameter,ge2018distributed}.

As these networks enter operational contexts, such as geodesy, clinical brain imaging, defence-grade gravimetry, the local field values $\theta_j$ become sensitive data. This would be particularly an issue in the context of health and biology~\cite{price2019privacy}. The fusion centre may legitimately require only the aggregate $\sum_jw_j\theta_j$, while individual parameters (a patient's neural field, an institution's
clock offset, or a site's gravitational anomaly) should remain private from adversarial curious-but-honest observers.

Classical sensor privacy has been studied via Fisher information privacy~\cite{farokhi2019fisher_sandberg}, information theoretic privacy~\cite{issa2020operational},
and differential privacy~\cite{wei2020federated}. The quantum setting raises new questions regarding the interplay between entanglement for Heisenberg-limited precision versus the privacy of individual sensor parameters (or a small subset). Particularly, it is paramount to investigate fundamental trade-off constraining precision and privacy. 

We answer this precisely via a quantum Fisher information (QFI) duality, establishing a conservation law on how quantum Fisher information is partitioned across orthogonal parameter directions under local unitary encoding. We show that our result implies co-optimality of precision and privacy, i.e., attaining Heisenberg-limited precision for the sensing target renders all alternative privacy-intrusive estimations impossible. However, a genuine tradeoff emerges when several independent targets must be estimated simultaneously.

Our QFI duality is structurally distinct from the incompatibility bounds of multi-parameter
metrology~\cite{ragy2016compatibility,demkowicz2020multi,gessner2018sensitivity}, which bound multi-target precision via measurement incompatibility. Our result derives from the trace identity $\tr(\cF_Q)=N$ imposed
by the \emph{local unitary encoding} structure, and not from measurement incompatibility, and are mapped into a \emph{precision-privacy} tradeoff. Quantum privacy metrics have started to appear more recently in the literature with a focus on developing quantum equivalents of differential privacy and information-theoretic privacy~\cite{farokhi2024maximal, farokhi2024barycentric,hirche2023quantum, farokhi2025sample, cheng2025invitation}. However, these results have not been applied to distributed quantum sensing.  

In this letter, we establish a QFI duality or conservation law for distributed quantum sensor networks
with local phase encoding, demonstrating that the sum of the Fisher information for two distinct target aggregate estimates is bounded by the number of the sensors. The bound is attained by  all equatorial states when there are only two sensors and by GHZ states generally.
Heisenberg-limited precision for one aggregate implies zero QFI for all other directions, ensuring privacy. This means no privacy-utility trade-off when the experiment is set up to estimate for only a single sensing target. However, the QFI duality presents a genuine tradeoff when several independent targets must be estimated simultaneously, i.e., Heisenberg-limited precision in two distinct directions is not attainable.

\section{Formulation}
Consider $N$ sensors, each equipped with a qubit space $\mathcal{H}\cong\mathbb{C}^2$. A central station or entity prepares a pure probe state $\ket{\psi_0}\in \mathcal{H}\otimes \cdots \otimes\mathcal{H}$ and distributes the $j$-th qubit
to sensor~$j$. Each sensor applies the local phase encoding
\begin{equation}
U_j(\theta_j)=\exp\left(-\frac{1}{2}i\theta_j\sigma_z\right),\quad\theta_j\in[0,2\pi),
\label{eq:encoding}
\end{equation}
where $\sigma_z$ is the Pauli $Z$ operator. This results in the encoded state
\begin{align}
    \ket{\psi(\bm{\theta})}
    &=(U_1(\theta_1)\otimes \cdots \otimes U_N(\theta_N))\ket{\psi_0}\\
    &=\exp\left(-\frac{1}{2} i\sum_{j=1}^N\theta_j\sigma_z^{(j)}\right)\ket{\psi_0}
\end{align}
where $\bm{\theta}=(\theta_1\dots,\theta_N)$ and $\sigma_z^{(j)}=I\otimes \cdots \otimes \sigma_z \otimes \cdots \otimes I $ with location of $\sigma_z$ operator aligning with the $j$-th qubit (i.e., Pauli $Z$ operator is applied to the $j$-th qubit). 

The quantum Fisher information matrix (QFIM) is 
\begin{align}
    [\cF_Q(\bm{\theta})]_{ij} = \frac{1}{2}\mathrm{Tr}\!\left[\rho_{\boldsymbol{\theta}}\,\{L_i, L_j\}\right],
\end{align}
where $L_i$ is the symmetric logarithmic derivative (SLD) with respect to $\theta_i$, defined implicitly as the solution to  $\partial \rho(\bm{\theta})/\partial \theta_i = \{L_i, \rho(\bm{\theta})\}/2$ with density operator $\rho(\bm{\theta}) = \ket{\psi(\bm{\theta})} \bra{\psi(\bm{\theta})}$ and anti-commutator $\{A,B\}=AB+BA$. Using~(27) in~\cite{liu2020quantum}, the QFIM can be simplified to
\begin{align}
[\cF_Q(\bm{\theta})]_{ij}=\langle \sigma_z^{(i)}\sigma_z^{(j)}\rangle
-\langle \sigma_z^{(i)}\rangle\langle \sigma_z^{(j)}\rangle,
\label{eq:qfim}
\end{align}
where $\langle A\rangle=\bra{\psi_0}A\ket{\psi_0}$. The quantum Fisher information (QFI) for any linear combination $\bm{u}^\top \bm{\theta}$ is
$\cF_Q(\bm{u}^\top \bm{\theta})=\bm{u}^\top\cF_Q\bm{u}$; see
Proposition~2.1 in~\cite{liu2020quantum}.

The fusion station or entity wants to estimate a 
target combination of parameters $\bm{\theta}$ denoted by $f=\hat{\bm{w}}^\top \bm{\theta}$, where $\hat{\bm{w}}$ is a normalised vector, i.e., $\|\hat{\bm{w}}\|_2=1$, capturing the weight of each parameter in the target combination. Vector $\hat{\bm{w}}$ may be referred to as the sensing direction. The quantum Cram\'{e}r-Rao bound in~\cite{HELSTROM1967101}
states that the variance of any unbiased estimator
$\hat{f}$ of $f$, i.e., $\mathbb{E}\{\hat{f}\}=f$,  satisfies
\begin{equation}
\Var(\hat{f})\geq\frac{1}{\cF_Q(\hat{\bm{w}}^\top \bm{\theta})}.
\label{eq:qcrb}
\end{equation}
Therefore, QFI provides a fundamental bound on quality of unbiased estimators. In this letter, we present a QFI duality, demonstrating that the sum of QFIs for any two orthogonal sensing directions is bounded. That is, we cannot obtain high-quality estimates in two directions simultaneously.  

\section{QFI Duality}
The probe state $\ket{\psi_0}$ is \emph{equatorial} if
$\langle \sigma_z^{(j)}\rangle=0$ for all $j$. Examples of equatorial probe states are the $N$-qubit GHZ state
$\ket{\Phi_N^+}=(\ket{0}\otimes\cdots \otimes\ket{0}+\ket{1}\otimes\cdots \otimes\ket{1})/\sqrt{2}$ and the separable
probe $\ket{+}^{\otimes N}$ with $\ket{\pm}=(\ket{0}\pm\ket{1})/\sqrt{2}$. Now, we are ready to present the main result of this letter.

\begin{theorem}
\label{thm:duality}
For any two unit vectors $\hat{\bm{w}}\perp \hat{\bm{v}}$:
\begin{equation}
\cF_Q(\hat{\bm{w}}^\top \bm{\theta})
+\cF_Q(\hat{\bm{v}}^\top \bm{\theta})\leq N.
\label{eq:duality}
\end{equation}
Equality holds with equatorial probe states for $N=2$.
\end{theorem}

\begin{proof}
We first show that $\tr(\cF_Q(\bm{\theta}))\leq N$. Setting $i=j$ in~\eqref{eq:qfim} gives $[\cF_Q(\bm{\theta})]_{ii}=\langle (\sigma_z^{(i)})^2\rangle-\langle (\sigma_z^{(i)})\rangle^2\leq\langle (\sigma_z^{(i)})^2\rangle =\langle I\rangle=1$. 
Therefore, $\tr(\cF_Q(\bm{\theta}))=\sum_{i=1}^N [\cF_Q(\bm{\theta})]_{ii}\leq N$. Proposition~2.1 in~\cite{liu2020quantum} shows that  $\cF_Q$ is a real symmetric positive semi-definite matrix. Therefore, it admits spectral decomposition $\cF_Q(\bm{\theta})=\sum_k\lambda_k\bm{e}_k\bm{e}_k^\top$ with $\lambda_k\geq 0$, $\bm{e}_k^\top \bm{e}_k=1$, and $\tr(\cF_Q(\bm{\theta}))=\sum_{k}\lambda_k\leq N$. For any unit vector $\hat{\bm{u}}$, we have $\cF_Q(\hat{\bm{u}}^\top \bm{\theta})=\hat{\bm{u}}^\top\cF_Q(\bm{\theta})\hat{\bm{u}}=\sum_k\lambda_k(\hat{\bm{u}}^\top\bm{e}_k)^2$. We can write each  eigenvector $\bm{e}_k$ as $\bm{e}_k=\alpha_k\hat{\bm{w}}+\beta_k\hat{\bm{v}}+\bm{r}_k,$ where $\alpha_k=\hat{\bm{w}}^\top \bm{e}_k$, $\beta_k=\hat{\bm{v}}^\top \bm{e}_k$, and $\bm{r}_k\perp\hat{\bm{w}},\hat{\bm{v}}$ (the component of $\bm{e}_k$ orthogonal to both $\hat{\bm{w}}$ and $\hat{\bm{v}}$, obtained by subtracting
their projections). We have $|\bm{e}_k|^2=\alpha_k^2+\beta_k^2+|\bm{r}_k|^2=1$, which results in $\alpha_k^2+\beta_k^2\leq 1$. Recalling the definitions of $\alpha_k$ and $\beta_k$, we get $(\hat{\bm{w}}^\top \bm{e}_k)^2+(\hat{\bm{v}}^\top \bm{e}_k)^2\leq 1$. Combining all these, we get $\cF_Q(\hat{\bm{w}}^\top \bm{\theta})
\!+\!\cF_Q(\hat{\bm{v}}^\top \bm{\theta})=\sum_{k} \lambda_k[(\hat{\bm{w}}^\top\bm{e}_k)^2\!+\!(\hat{\bm{v}}^\top\bm{e}_k)^2]\leq \sum_{k}\lambda_k=N.$

We now show the equality for equatorial pure probe states and $N=2$. Note that $\langle \sigma_z^{(i)}\rangle=0$, which implies that $[\cF_Q(\bm{\theta})]_{ii}=\langle (\sigma_z^{(i)})^2\rangle =\langle I\rangle=1$. Therefore, $\tr(\cF_Q(\bm{\theta}))= 2$. The rest follows from that $\cF_Q(\hat{\bm{w}}^\top \bm{\theta})
+\cF_Q(\hat{\bm{v}}^\top \bm{\theta})= \hat{\bm{w}}^\top \cF_Q(\bm{\theta})\hat{\bm{w}}+\hat{\bm{v}}^\top \cF_Q(\bm{\theta})\hat{\bm{v}}= \tr( \cF_Q(\bm{\theta})\hat{\bm{w}}\hat{\bm{w}}^\top) +\tr( \cF_Q(\bm{\theta})\hat{\bm{v}}\hat{\bm{v}}^\top)
= \tr( \cF_Q(\bm{\theta})(\hat{\bm{w}}\hat{\bm{w}}^\top+\hat{\bm{v}}\hat{\bm{v}}^\top))= \tr( \cF_Q(\bm{\theta}))=2$. Note that $\hat{\bm{w}}\hat{\bm{w}}^\top+\hat{\bm{v}}\hat{\bm{v}}^\top=I$, which holds because, in $N=2$, any two orthonormal vectors form a complete basis.
\end{proof}

For $N=2$, the QFI duality in~\eqref{eq:duality}, which holds as equality, is a conservation law of information. The sum of QFI in
any two independent parameter combinations is fixed for all equatorial states, regardless of the entanglement structure. For $N>2$, this relationship can be interpreted as a precision-privacy complementarity relation. 

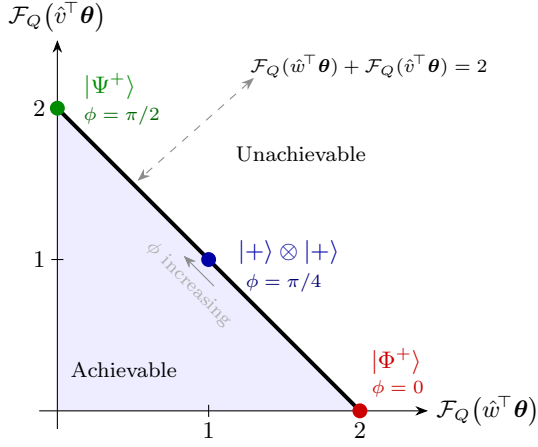
\begin{figure}
    \centering
    \begin{tikzpicture}[x=2.0cm,y=2.0cm,
      >=Stealth,font=\small]
      \fill[blue!7!white] (0,0) -- (2,0) -- (0,2) -- cycle;
      \draw[line width=1.5pt,black] (2,0) -- (0,2);
      \draw[->, thin] (-0.12,0) -- (2.45,0)
        node[right, inner sep=3pt]{$\mathcal{F}_Q\!\left(\hat{w}^{\!\top}\bm{\theta}\right)$};
      \draw[->, thin] (0,-0.12) -- (0,2.45)
        node[above, inner sep=3pt]{$\mathcal{F}_Q\!\left(\hat{v}^{\!\top}\bm{\theta}\right)$};
      \foreach \val in {1, 2}{
        \draw[thin] (\val,  0.04) -- (\val, -0.04)
          node[below, inner sep=2pt] {$\val$};
        \draw[thin] (0.04, \val)  -- (-0.04, \val)
          node[left, inner sep=2pt] {$\val$};
      }
      \fill[red!80!black] (2,0) circle (2.8pt);
      \node[above left=2pt, red!80!black]
        at (2.5, 0.15) {$|\Phi^+\rangle$};
      \node[below=-15pt, font=\scriptsize, red!80!black]
        at (2.25, 0.0) {$\phi=0$};
      \fill[blue!65!black] (1,1) circle (2.8pt);
      \node[right=5pt, blue!65!black]
        at (1.05, 1.05) {$\ket{+}\otimes\ket{+}$};
      \node[right=5pt, font=\scriptsize, blue!45!black]
        at (1.1, 0.85) {$\phi=\pi/4$};
      \fill[green!55!black] (0,2) circle (2.8pt);
      \node[right=5pt, green!55!black]
        at (0.03, 2.15) {$|\Psi^+\rangle$};
      \node[right=5pt, font=\scriptsize, green!40!black]
        at (0.03, 1.95) {$\phi=\pi/2$};
      \node[black, font=\footnotesize]
        at (0.44, 0.27) {Achievable};
      \node[black, font=\footnotesize]
        at (1.61, 1.70) {Unachievable};
      \node[
        rotate = 0,
        above  = 5pt,
        font   = \scriptsize,
        black,
      ] at (2.06, 2.06)              {$\mathcal{F}_Q(\hat{w}^{\!\top}\bm{\theta})+\mathcal{F}_Q(\hat{v}^{\!\top}\bm{\theta}) = 2$};
      \draw[<->,dashed,black!45] (1.3,2.2) -- (.53,1.52);
      \node[
        rotate = -45,
        below  = 4pt,
        font   = \scriptsize,
        gray!62,
      ] at (1.0, 0.97) {$\phi$ increasing};    
      \draw[<-,black!45,xshift=-1.3mm,yshift=-1.5mm] (.9,1.1) -- (1.1,0.9);
    \end{tikzpicture}
    \caption{Precision-privacy Pareto frontier for $N=2$ equatorial probes, with $\hat{\bm{w}}=\begin{bmatrix} 1/\sqrt{2} & 1/\sqrt{2}\end{bmatrix}$ targeting the sum $(\theta_1+\theta_2)/\sqrt{2}$ and $\hat{\bm{v}}=\begin{bmatrix} 1/\sqrt{2} & -1/\sqrt{2}\end{bmatrix}$ targeting the difference $(\theta_1-\theta_2)/\sqrt{2}$. Every equatorial two-qubit state lies exactly on the conservation-law line $\mathcal{F}_Q(\hat{\bm{w}}^{\top} \bm{\theta})+\mathcal{F}_Q( \hat{\bm{v}}^{\top}\bm{\theta})=2$ (Theorem~\ref{thm:duality}, $N=2$ equality); none lies strictly inside it. Three states from the Bell-state family $\cos(\phi)\ket{\Phi^+} +\sin(\phi)\ket{\Psi^+}$ are highlighted (arrow: direction of increasing~$\phi$): the Bell state $\ket{\Phi^+}$ (red, $\phi=0$) at the co-optimal corner $(2,0)$, where Heisenberg-limited precision for the sum coexists with zero QFI for the difference, making individual parameter recovery impossible; the separable state $\ket{+}\otimes\ket{+}$ (blue, $\phi=\pi/4$) at the midpoint $(1,1)$, with equal QFI in both directions; and the Bell state $\ket{\Psi^+}$ (green, $\phi=\pi/2$) at $(0,2)$, which is maximally sensitive to the difference but blind to the sum. For $N>2$, the frontier shifts outward to $\mathcal{F}_Q(\hat{\bm{w}}^{\top} \bm{\theta})+\mathcal{F}_Q(\hat{\bm{v}}^{\top}\bm{\theta})=N$.}
    \label{fig:pareto}
\end{figure}

Figure~\ref{fig:pareto} illustrates Theorem~\ref{thm:duality} for
$N=2$ with $\hat{\bm{w}}=\begin{bmatrix}
1/\sqrt{2} & 1/\sqrt{2}\end{bmatrix}$ and $\hat{\bm{v}}=\begin{bmatrix}
1/\sqrt{2} & -1/\sqrt{2}\end{bmatrix}$. We can parameterise equatorial two-qubit states as $\ket{\psi_\phi}= \cos(\phi)\ket{\Phi^+}+ \sin(\phi)\ket{\Psi^+},$ for all $\phi\in[0,\pi/2],$
where $\ket{\Phi^+}=(\ket{00}+\ket{11})/\sqrt{2}$ and
$\ket{\Psi^+}=(\ket{01}+\ket{10})/\sqrt{2}$ are Bell states.
Both are equatorial, mutually orthogonal, and their superposition
is normalised for all~$\phi$.
The local Pauli--$Z$ operators satisfy
$\sigma_z^{(1)}\sigma_z^{(2)}\ket{\Phi^+} = +\ket{\Phi^+}$ and
$\sigma_z^{(1)}\sigma_z^{(2)}\ket{\Psi^+} = -\ket{\Psi^+}$,
so the off-diagonal entry of the QFIM~\eqref{eq:qfim} evaluates to $[\cF_Q(\bm{\theta})]_{ij}=\langle\sigma_z^{(1)}\sigma_z^{(2)}\rangle= \cos^2\phi - \sin^2\phi = \cos 2\phi.$ Therefore, $\mathcal{F}_Q(\hat{\bm{w}}^\top \bm{\theta})= 1 + \cos 2\phi$ and $\mathcal{F}_Q(\hat{\bm{v}}^\top \bm{\theta})= 1 - \cos 2\phi.$ Evidently, $\mathcal{F}_Q(\hat{\bm{w}}^\top \bm{\theta})+\mathcal{F}_Q(\hat{\bm{v}}^\top \bm{\theta})=2$ for every~$\phi$, confirming the conservation law in Theorem~\ref{thm:duality}.
As $\phi$ increases from $0$ to $\pi/2$, the operating point traces
the \emph{complete} Pareto frontier, the straight line capturing the boundary of the achievable region for the quantum Fisher information, passing
through three special cases:
\begin{itemize}
  \item \emph{Bell state} $\ket{\Phi^+}$ ($\phi=0$):
        Heisenberg-limited precision for the sum
        $\hat{\bm{w}}^\top\bm{\theta}=(\theta_1+\theta_2)/\sqrt{2}$
        while 
        $\hat{\bm{v}}^\top\bm{\theta}=(\theta_1-\theta_2)/\sqrt{2}$
        is completely impossible to be estimated ($\mathcal{F}_Q(\hat{\bm{v}}^\top \bm{\theta})=0$), making
        individual parameter recovery impossible, but the target estimation attains the highest achievable quality.
  \item \emph{Separable state} $\ket{+}\otimes\ket{+}$
        ($\phi=\pi/4$): Equal QFI in
        both directions, implying neither Heisenberg-limited nor
        individual privacy. 
  \item \emph{Bell state} $\ket{\Psi^+}$ ($\phi=\pi/2$): Maximally sensitive to the difference $\hat{\bm{v}}^\top \bm{\theta}=(\theta_1-\theta_2)/\sqrt{2}$ but blind to the sum $\hat{\bm{w}}^\top\bm{\theta}=(\theta_1+\theta_2)/\sqrt{2}$.
\end{itemize}
Every equatorial two-qubit state lies on this frontier; none lies
strictly inside it. This verifies the equality in Theorem~\ref{thm:duality} for $N=2$. 

\begin{corollary}[GHZ optimality]
\label{cor:coopt}
Let $\hat{\bm{w}}=\bm{1}/\sqrt{N}$ with $\bm{1}=\begin{bmatrix}
    1 & \cdots & 1
\end{bmatrix}$. 
The $N$-qubit GHZ probe admits \begin{equation}
\cF_Q(\hat{\bm{w}}^\top \bm{\theta})=N,\qquad
\cF_Q(\hat{\bm{v}}^\top \bm{\theta})=0
\quad\forall\,\hat{\bm{v}}\perp\hat{\bm{w}}.
\label{eq:ghzqfi}
\end{equation}
GHZ is the unique equatorial probe (up to local unitaries commuting
with the $\sigma_z$) achieving~\eqref{eq:ghzqfi}.
\end{corollary}

\begin{proof}
First, $\sigma_z^{(i)}\sigma_z^{(j)}\ket{0}^{\otimes N}=\ket{0}^{\otimes N}$ and $\sigma_z^{(i)}\sigma_z^{(j)}\ket{1}^{\otimes N}=\ket{1}^{\otimes N}$. Therefore, $\sigma_z^{(i)}\sigma_z^{(j)}\ket{\Phi_N^+}=\ket{\Phi_N^+}$. Using this property while noting that GHZ is equatorial, we get $[\cF_Q(\bm{\theta})]_{ij}=1$ for all $i,j$. This implies that  $\cF_Q(\bm{\theta})=\bm{1}\bm{1}^\top$. As a result, $\cF_Q(\hat{\bm{w}}^\top \bm{\theta})=(\bm{1}^\top\hat{\bm{w}})^2=N$. Similarly, $\cF_Q(\hat{\bm{v}}^\top \bm{\theta})=(\bm{1}^\top\hat{\bm{v}})^2=0$ because $\hat{\bm{v}}\perp \bm{1}$. 

Now, we prove the uniqueness. Suppose $\cF_Q(\hat{\bm{v}}^\top \bm{\theta})=0$ for all $\hat{\bm{v}}\perp \bm{1}$.
Then, for any such $\hat{\bm{v}}$, $\cF_Q(\bm{\theta})\hat{\bm{v}}=0$, i.e., $\hat{\bm{v}}$
is in the null space of $\cF_Q(\bm{\theta})$ for every $\hat{\bm{v}}\perp \bm{1}$. The null space of $\cF_Q(\bm{\theta})$ therefore contains the $(N-1)$-dimensional
subspace $\{\bm{1}\}^\perp$, so $\cF_Q(\bm{\theta})=c\bm{1}\bm{1}^\top$ for some $c\geq 0$.
From, the proof of Theorem~\ref{thm:duality}, we know that, for any equatorial probe state, $\tr(\cF_Q)=N$. Hence, $c=1$. This is produced by the GHZ state and other states that follow from manipulating GHZ state via local unitaries that commute with the encoding~\eqref{eq:encoding} to keep $\langle \sigma_z^{(i)} \sigma_z^{(j)}\rangle$ the same.
\end{proof}

Corollary~\ref{cor:coopt} reveals an interesting quantum property. GHZ provides a perfect probe state for estimating $\hat{\bm{w}}^\top \bm{\theta}$ with $\hat{\bm{w}}=\bm{1}/\sqrt{N}$ while no other target $\hat{\bm{v}}^\top \bm{\theta}$ ($\hat{\bm{v}}\perp \hat{\bm{w}}$) can be estimated. That is, if the experiment is set to estimate $\hat{\bm{w}}^\top \bm{\theta}$, any curious (in terms of estimating $\hat{\bm{v}}^\top \bm{\theta}$) but honest (captured by not changing the probe state) estimator cannot simultaneously estimate $\hat{\bm{v}}^\top \bm{\theta}$.

\begin{corollary}[Separable suboptimality]
\label{cor:sep}
For separable probe state $\ket{+}^{\otimes N}$ and any two unit vectors $\hat{\bm{w}}\perp \hat{\bm{v}}$:
\begin{equation}
\cF_Q(\hat{\bm{w}}^\top \bm{\theta})=
\cF_Q(\hat{\bm{v}}^\top \bm{\theta})=1.
\label{eq:sepqfi}
\end{equation}
\end{corollary}

\begin{proof}
For $\ket{+}^{\otimes N}$, $\langle \sigma_z^{(i)} \sigma_z^{(j)}\rangle=0$ if $i\neq j$. Using this, and again noting that $\ket{+}^{\otimes N}$ is equatorial, we get that $\cF_Q(\bm{\theta})=I$. As a result, $\cF_Q(\hat{\bm{w}}^\top \bm{\theta})=\hat{\bm{w}}^\top\hat{\bm{w}} =1$ and $\cF_Q(\hat{\bm{v}}^\top \bm{\theta})=\hat{\bm{v}}^\top \hat{\bm{v}}=1$.
\end{proof}

Corollary~\ref{cor:sep} shows that, for the separable probe $\ket{+}^{\otimes N}$,
$\cF_Q(\hat{\bm{w}}^\top \bm{\theta})+
\cF_Q(\hat{\bm{v}}^\top \bm{\theta})=2$ for all $N$. That is, the separable probe state would not saturate the bound in Theorem~\ref{thm:duality}, except for $N=2$. Separable states are not ideal for achieving supremacy in distributed quantum sensing. 

Corollaries~\ref{cor:coopt} and~\ref{cor:sep} together show that GHZ probe states beat separable ones not only on precision for estimating $\hat{\bm{w}}^\top \bm{\theta}$ (by a factor of $N$), but also in terms of  precision-privacy as $\cF_Q(\hat{\bm{v}}^\top \bm{\theta})=1$ for separable probes while $\cF_Q(\hat{\bm{v}}^\top \bm{\theta})=0$ for GHZ probes.

When $r\geq 2$ independent functions must be estimated
simultaneously, a genuine precision tradeoff emerges.
For two sensing targets $f_1=\hat{\bm{w}}^\top \bm{\theta}$ and $f_2=\hat{\bm{v}}^\top\bm{\theta}$
with orthogonal sensing directions $\bm{v}\perp\bm{w}$, Theorem~\ref{thm:duality}
gives $\mathcal{F}_Q(\hat{\bm{w}}^\top \bm{\theta})+\mathcal{F}_Q(\hat{\bm{v}}^\top\bm{\theta})\leq N$. Thus, for any probe state achieving $\mathcal{F}_Q(\hat{\bm{v}}^\top\bm{\theta})\geq\delta>0$, it  must also hold that $\mathcal{F}_Q(\hat{\bm{w}}^\top \bm{\theta})\leq N-\delta.$
Heisenberg scaling in one direction therefore forces all other directions unestimable. Individual parameter recovery is a special case because estimating $\theta_j$ requires positive QFI for direction $\hat{\bm{v}}_j=(\bm{e}_j-(\bm{e}_j^\top \hat{\bm{w}})\hat{\bm{w}})/\sqrt{1-1/N}$. However, this observation clearly extends to any two distinct targets. 

\section{Discussion and Conclusion}

The QFI duality (Theorem~\ref{thm:duality}) is a structural consequence of local unitary encoding. The trace identity $\tr(\cF_Q(\bm{\theta}))=N$ fixes the total Fisher information across all directions. Concentrating it into the sensing direction $\hat{\bm{w}}$ (Heisenberg scaling) depletes it from all directions (e.g., individual privacy). The GHZ probe achieves this concentration uniquely, co-optimising precision and privacy for single-target sensing. GHZ encodes only $\sum_j\theta_j$ because its rank-1 QFIM
$\bm{1}\bm{1}^\top$ concentrates all the Fisher information in the sum direction, enabling protection of individual parameters.

\textit{Physical implications.} In atomic clock
networks~\cite{komar2014quantum}, GHZ operation already required for Heisenberg-limited synchronisation simultaneously protects each institution's clock offsets.
In NV-centre brain imaging~\cite{barry2016optical}, the same entangled-probe design that maximises sensitivity protects per-sensor patient data under regulatory privacy requirements. In quantum gravimeter networks~\cite{stray2022quantum}, individual-site
gravitational measurements remain private to an adversary. In each case, the privacy guarantee is a
free byproduct of the Heisenberg-scaling architecture.

\textit{Conclusion.} We have proved a QFI duality that formally unifies precision and privacy in distributed quantum sensing. Heisenberg-limited precision for a single sensing target implies complete individual parameter privacy for that target. GHZ probes uniquely co-optimise precision and privacy, while separable probes fail to reach the Pareto frontier for more than two sensors. For multiple simultaneous targets, a quantitative tradeoff with a closed-form Pareto boundary emerges. These results reframe
entanglement in quantum sensor networks: it is not merely a resource for beating the standard quantum limit, but the mechanism by which precision and privacy are simultaneously optimised. Privacy for free is not present in classical frameworks~\cite{kifer2011no}.

\bibliography{ref}

\end{document}